\documentclass[a4paper,12pt,notitlepage, twocolumn]{article}
\usepackage[utf8]{inputenc}
\usepackage[T1]{fontenc}
\usepackage{fullpage}
\usepackage{times}
\usepackage{amsmath}
\usepackage{amsfonts}
\usepackage{amssymb}
\usepackage[makeroom]{cancel}
\usepackage{amsthm}
\usepackage{graphicx}
\usepackage{float}
\usepackage{multirow}
\usepackage{multicol}
\usepackage{cuted}

\usepackage{mathtools}

\usepackage{bbm}
\usepackage{braket}
\usepackage{ytableau,tikz,varwidth}
\usepackage[enableskew]{youngtab}
\usepackage{hyphenat}
\usepackage{verbatim}
\usepackage{subcaption}
\usepackage{hyperref}
\usepackage{soul}
\usepackage{etoolbox}
\usepackage{nameref}
\usepackage{array,multirow,graphicx}
\usetikzlibrary{tikzmark,decorations.pathreplacing,calc}
\usepackage{changepage}
\usepackage{enumitem}

\usepackage[style=numeric-comp,sorting=none,sortcites=true]{biblatex}

\AtEveryBibitem{\clearfield{doi}}
\AtEveryBibitem{\clearfield{issn}}
\AtEveryBibitem{\clearfield{isbn}}
\AtEveryBibitem{\clearfield{pages}}
\AtEveryBibitem{\clearfield{issue}}
\AtEveryBibitem{\clearfield{url}}
\AtEveryBibitem{\clearfield{month}}

\usepackage{listings}
\usepackage{xcolor}
\definecolor{codegreen}{rgb}{0,0.6,0}
\definecolor{codegray}{rgb}{0.5,0.5,0.5}
\definecolor{codepurple}{rgb}{0.58,0,0.82}
\definecolor{backcolour}{rgb}{0.95,0.95,0.92}
\lstdefinestyle{mystyle}{
    backgroundcolor=\color{white},   
    commentstyle=\color{codegreen},
    keywordstyle=\color{magenta},
    numberstyle=\tiny\color{codegray},
    stringstyle=\color{codepurple},
    basicstyle=\ttfamily\footnotesize,
    breakatwhitespace=false,         
    breaklines=true,                 
    captionpos=b,                    
    keepspaces=true,                 
    numbers=left,                    
    numbersep=5pt,                  
    showspaces=false,                
    showstringspaces=false,
    showtabs=false,                  
    tabsize=2
}
\graphicspath{{plots/}}

\theoremstyle{definition}

\allowdisplaybreaks
\DeclareRobustCommand{\bbone}{\text{\usefont{U}{bbold}{m}{n}1}}

\def\CC {{\mathbb C}}     
\def\PP {{\mathbb P}}     
\def\RR {{\mathbb R}}     

\def\mc {\mathcal}

\newcommand{\tr}{\mbox{Tr}}

\DeclarePairedDelimiterX\bk[2]{\langle}{\rangle}{#1 \delimsize\vert #2}
\DeclarePairedDelimiterX\kb[2]{\vert }{\vert }{#1 \delimsize\rangle\langle#2}

\def\={\;=\;} \def\+{\,+\,}

\newtheorem{theorem}{Theorem}

\newtheorem{definition}[theorem]{Definition}

\usepackage{tikz}

\apptocmd\appendix{%
  \addcontentsline{toc}{chapter}{Appendix}%
  \counterwithin{equation}{section}%
  \counterwithin{figure}{section}%
  \counterwithin{table}{section}%
}{}{}

\begin{document}

\title{Symplectic Structures in Quantum Entanglement}
\author{Piotr Dulian$^{1, 2}$ and Adam Sawicki $^{1,3}$}

\date{%
    $^1$ Center for Theoretical Physics, Polish Academy of Sciences,\\ Al. Lotnik\'ow 32/46, 02-668 Warszawa, Poland,\\
    $^2$ Centre for Quantum Optical Technologies, Centre of New Technologies,\\ University of Warsaw, Banacha 2c, 02-097 Warszawa, Poland\\
    $^3$ Guangdong Technion - Israel Institute of Technology, 241 Daxue Road, Jinping District, Shantou, Guangdong Province, China
 }
\maketitle

\begin{abstract}
In this work, we explore the implications of applying the formalism of symplectic geometry to quantum mechanics, particularly focusing on many-particle systems. We extend the concept of a symplectic indicator of entanglement, originally introduced by Sawicki et al. \cite{sawicki2011}, to these complex systems. Specifically, we demonstrate that the restriction of the symplectic structure to manifolds comprising all states characterized by isospectral reduced one-particle density matrices, \( M_{\mu(\psi)}^0 \), exhibits degeneracy for non-separable states. We prove that the degree of degeneracy at any given state \( \ket{\varphi} \in M_{\mu(\psi)}^0 \) corresponds to the degree of degeneracy of the symplectic form \( \omega \) when restricted to the manifold of states that are locally unitary equivalent with \( \ket{\varphi} \). Additionally, we provide a physical interpretation of this symplectic indicator of entanglement, articulating it as an inherent ambiguity within the associated classical dynamical framework. Our findings underscore the pivotal role of symplectic geometry in elucidating entanglement properties in quantum mechanics and suggest avenues for further exploration into the geometric structures underlying quantum state spaces.

\end{abstract}

\section{Introduction}
From the very inception of quantum mechanics, quantum entanglement \cite{horodecki2009} has emerged as a crucial feature that distinguishes it from classical theories \cite{einstein1935, schrodinger1935}. Various formalisms have been proposed to articulate this difference, including Bell inequalities \cite{bell1964} and entropic inequalities \cite{cerf1997, horodecki1994}, alongside numerous methods for detecting quantum entanglement. Techniques such as positive partial transpose \cite{peres1996a}, positive maps \cite{horodecki1996a, choi1972}, and entanglement witnesses \cite{terhal2000a, terhal2000b} have been developed, as well as methods for quantifying entanglement \cite{vedral1998, vedral1997b}. To this day, quantum entanglement remains a vibrant field of research, with applications spanning communication \cite{brukner2004, buhrman2001a, cleve1997, bennett1996}, spatial orientation \cite{bovino2006b, brukner2005}, improvements in frequency standards \cite{giovannetti2004, huelga1997, wineland1992}, clock synchronization \cite{jozsa2000}, and cryptography \cite{collins2002, gisin1999, gisin2000}.

In this work, we investigate the implications of applying the formalism of symplectic geometry, commonly utilized in classical mechanics, to the realm of quantum mechanics developed in \cite{dirac1972, strocchi1966, Weinberg1989, cantoni1982, heslot1985, BRODY2001, anandan1990, cirelli1990, GIBBONS1992}. The symplectic structure of the state space has proven to be particularly valuable in addressing local unitary equivalence issues \cite{SWK13, SK11}, in the classification of states concerning Stochastic Local Operations and Classical Communication \cite{SOK12,SOK14,MS18}, in the study of Berry phases \cite{page1987, samuel1988, benedict1989} or twistor theory \cite{hughston1995}.

For a finite-dimensional Hilbert space \(\mathcal{S}\), the space of pure states is represented as a complex projective space, \(\mathbb{P}(\mathcal{S})\), which is endowed with a natural symplectic structure defined via the inner product. Consequently, for any quantum Hamiltonian \(\hat{H}:\mathcal{S} \rightarrow \mathcal{S}\), one can define a function on \(\mathbb{P}(\mathcal{S})\), which we refer to as the classical Hamiltonian function corresponding to \(\hat{H}\):

\[
\mathcal{H}: \mathbb{P}(\mathcal{S}) \ni \ket{\psi} \mapsto \bra{\psi}\hat{H}\ket{\psi} \in \mathbb{R}.
\]

The integral lines of the corresponding Hamiltonian vector field \(V_{\mathcal{H}}\) precisely represent the solutions of the Schrödinger equation associated with \(\hat{H}\).

To encapsulate quantum correlations within this framework, we consider a multipartite system \(\mathcal{S} = \mathcal{S}_{1} \otimes \ldots \otimes \mathcal{S}_{L}\), composed of \(L\) distinguishable subsystems, each governed by non-interacting Hamiltonians \(\hat{H}\), which we refer to as local. The corresponding classical local Hamiltonians are functions of the reduced one-particle density matrices. Given that quantum correlations are invariant under the quantum evolution generated by local Hamiltonians, it is natural to partition the space \(\mathbb{P}(\mathcal{S})\) into submanifolds comprising states with isospectral reduced one-particle density matrices. These submanifolds, generally collections of the orbits of \(\mathbb U(\mathcal{S}_1)\times \ldots \times \mathbb U(\mathcal{S}_L)\) in \(\mathbb{P}(\mathcal{S})\), may not exhibit a symplectic structure. If this is the case, for any local \(\hat{H}_{\text{loc}}\), the associated Hamiltonian function \(\mathcal{H}_{\text{loc}}(\psi) = \bra{\psi}\hat{H}_{\text{loc}}\ket{\psi}\) does not uniquely define the Hamiltonian vector field, but rather one up to so-called null vector fields. Notably, among the orbits of \(\mathbb U(\mathcal{S}_1)\times \ldots \times  \mathbb U(\mathcal{S}_L)\), there exists precisely one that is symplectic: the orbit of separable states \cite{sawicki2011}. Thus, we can interpret entanglement as an ambiguity within the classical dynamics corresponding to local quantum Hamiltonians. The case of two-partite systems (\(L=2\)) was previously examined in \cite{sawicki2011}. This setting is particularly straightforward, as the submanifolds corresponding to states with isospectral reduced one-particle density matrices are exactly the orbits of \(\mathbb U(\mathcal{S}_1)\times \mathbb U(\mathcal{S}_2)\), and the formula for the degeneracy of \(\omega\) restricted to such an orbit was derived in \cite{sawicki2011} using the Kostant-Sternberg theorem \cite{kostant82}. 
In this paper, we expand our focus to the multipartite case, wherein a manifold of states possessing isospectral reduced one-particle density matrices is indeed a collection of \(\mathbb U(\mathcal{S}_1)\times \ldots \times \mathbb U(\mathcal{S}_L)\) orbits. By employing symplectic reduction techniques in Theorem \ref{thm:main}, which is our main result, we derive a formula for the degeneracy of \(\omega\) restricted to such a manifold. This provides further interpretation of the notion of symplectic form degeneracy as an indicator of entanglement, extending our understanding into the multipartite scenario.

The paper is organised as follows. In Section \ref{sec:class_mech} we review a mathematical formalism of classical mechanics using the notion of symplectic manifolds. Moreover, Section \ref{sec:presymplectit} discusses a crucial notion of presymplectic manifolds and describes how one defines classical equations of motion in this case. Next, in Section \ref{sec:quant_mech} we show how this formalism can be applied in quantum scenario and we discuss how this symplectic structure interacts with local dynamics. In Section \ref{sec:bipartite} we explain how it can be related to quantum entanglement by going through an example of a bipartite system. Finally, in section \ref{sec:multipartite} we generalise our result to a multipartite case.

\section{Classical mechanics as a symplectic theory}\label{sec:class_mech}
In classical mechanics one considers a phase space $\mathbb{R}^{2n}$ of position and momenta coordinates that are traditionally denoted as
\begin{equation}
    \mathbf{x} = (\mathbf{q}, \mathbf{p}) = (q^1, ..., q^n, p_1, ..., p_n).
\end{equation}
The dynamics is given by a Hamiltonian function:
\begin{equation}
    \mathcal{H}: \mathbb{R}^{2n} \rightarrow \mathbb{R},
\end{equation}
and Hamilton's equations:
\begin{equation}
    \dot{q}^i = \frac{\partial \mathcal{H}}{\partial p_i}, \quad \dot{p}_i = -\frac{\partial \mathcal{H}}{\partial q^i}.
\end{equation}
It is well known that classical mechanics can be expressed more formally as a theory that lives on a smooth manifold $M = \mathbb{R}^{2n}$ equipped with a bilinear map on Hamiltonian functions:
\begin{equation}
    \{\mathcal{F}, \mathcal{G}\} = \sum_{i=1}^n \left( \frac{\partial \mathcal{F}}{\partial q^i} \frac{\partial \mathcal{G}}{\partial p_i} - \frac{\partial \mathcal{F}}{\partial p_i} \frac{\partial \mathcal{G}}{\partial q^i} \right),
\end{equation}
called Poisson bracket. For any Hamiltonian function $\mathcal{H}$ we can construct a vector field $V_\mathcal{H}$:
\begin{equation}
    V_\mathcal{H} = \{ \cdot, \mathcal{H} \} = \sum_{i=1}^n \left( \frac{\partial \mathcal{H}}{\partial p_i} \frac{\partial}{\partial q^i} - \frac{\partial \mathcal{H}}{\partial q^i} \frac{\partial}{\partial p_i} \right).
\end{equation}
All vector fields that arise in this way are called Hamiltonian vector fields. Of course the vector filed $V_\mathcal{H}$ can be interpreted either as a differentiation or as a function from $M$ to the tangent bundle - $TM$. Using the second interpretation Hamilton's equations can be expressed as:
\begin{equation}
    \dot{\mathbf{x}} = V_\mathcal{H}(\mathbf{x}).
\end{equation}
In other words, solutions of Hamilton's equations are integral curves of vector field - $V_\mathcal{H}$, which can be concisely expressed using the notion of the vector flow of this field - $\Phi_{V_\mathcal{H}}$:
\begin{equation}
    \mathbf{x}(t) = \Phi_{V_\mathcal{H}}(t) \mathbf{x}_0.
\end{equation}

To check whether a given vector field
\begin{equation}
    V = \sum_{i=1}^n A^i \frac{\partial}{\partial q^i} + B_i \frac{\partial}{\partial q_i},
\end{equation}
is Hamiltonian it is useful to employ the symplectic form:
\begin{equation}\label{eq:canonic_sympl_form}
    \omega = \sum_{i=1}^n dp_i \wedge dq^i,
\end{equation}
in the following way. Consider 1-form:
\begin{equation}\label{eq:alpha1}
    \alpha = \omega(V, \cdot) = \sum_{i=1}^n \left( B_i dq^i - A^i dp_i \right).
\end{equation}
Since $M$ is contractible if $d\alpha = 0$ then there exists a Hamiltonian $\mathcal{H}$ such that
\begin{equation}\label{eq:alpha2}
    \alpha = -d\mathcal{H} = -\sum_{i=1}^n \left( \frac{\partial\mathcal{H}}{\partial q^i} dq^i + \frac{\partial\mathcal{H}}{\partial p_i} dp_i \right).
\end{equation}
Comparing \eqref{eq:alpha1} and \eqref{eq:alpha2} we obtain:
\begin{equation}
    A^i = \frac{\partial\mathcal{H}}{\partial p_i}, \quad B_i = -\frac{\partial\mathcal{H}}{\partial q^i},
\end{equation}
thus $V = \{\cdot, \mathcal{H}\}$. Analogously one can show that for Hamiltonian vector field $V_\mathcal{H}$ we have $d(\omega(V_\mathcal{H}, \cdot)) = 0$. Moreover, one can easily check that
\begin{equation}
    \{ \mathcal{F}, \mathcal{G}\} = -\omega(V_\mathcal{F}, V_\mathcal{G}).
\end{equation}
Thus symplectic form $\omega$ gives us the whole information about the relation between Hamiltionians and Hamiltonian vector fields and as such can be treated as a corner stone of classical mechanics theory. This leads to the introduction of abstract symplectic manifolds:
\begin{definition}\label{def:symplectic_manifold}
A symplectic manifold is a smooth manifold $M$ equipped with differential 2-form $\omega$ which is closed, $d\omega = 0$, and nondegenerate:
\begin{equation}
    \forall_u \; \omega(v, u) = 0 \;\Rightarrow\; v = 0.
\end{equation}    
Such a form is called symplectic form on $M$.
\end{definition}
The closeness of $\omega$ is needed for Poisson bracket to satisfy Jacobi identity and nondegeneration ensures that for any Hamiltonian function $\mathcal{H}$ one can uniquely choose related vector field as the one that satisfies $-d\mathcal{H} = \omega(V, \cdot)$. Furthermore, those two conditions combined by Darboux theorem \cite{arnold2013} allow for local choice of coordinates $(\mathbf{q}, \mathbf{p})$ on $M$ such that $\omega$ takes the form as in \eqref{eq:canonic_sympl_form}.


\subsection{Presymplectic manifolds}\label{sec:presymplectit}

It will be beneficial for the discussion in subsequent sections to investigate the implications of allowing the symplectic form \(\omega\) to be degenerate while maintaining its closed property. In this scenario, the manifold \(M\) is referred to as presymplectic, and the form \(\omega\), designated as the presymplectic form, fails to establish a one-to-one correspondence between \(1\)-forms and vector fields. To elucidate the consequences of this lack of isomorphism, we introduce the null tangent space at a point \(x \in M\)
\begin{equation}
    D_x = \{ v \in T_x M \,|\, \forall_u \omega(v, u) = 0 \},
\end{equation}
and the null distribution on $M$
\begin{equation}
    D = \bigcup_{x\in M} \{x\} \times D_x \subset TM.
\end{equation}
If $D$ is regular, that is the dimension of $D_x$ is the same for all $x\in M$, then there are smooth vector fields on $M$ called null that for any $x \in M$ belong to $D_x$. In this case, using Frobenious theorem, one can also show that the distribution $D$ is integrable \cite{bursztyn2013}. If $d\mathcal{H}$ lays outside of the image of
\begin{equation}\label{nutka-presym}
    V \mapsto \omega(V, \cdot).
\end{equation}
then there is no Hamiltonian vector field corresponding to $\mc H$ and we cannot consider Hamilton's equations of motion. This happens exactly when $d\mc H(D)\neq 0$ \cite{bursztyn2013}.

On the other hand, when $d\mc H$ belongs to the image of \eqref{nutka-presym} there are multiple vector fields corresponding to $d\mathcal{H}$. To see this assume  $V_\mathcal{H}$ is a Hamiltonian vector field of $\mathcal{H}$ and $V_0$ is null then $V'_\mathcal{H} = V_\mathcal{H} + V_0$ is also a Hamiltonian vector field of $\mathcal{H}$ as:
\begin{gather}
    \omega(V_\mathcal{H}', \cdot) = \omega(V_\mathcal{H}, \cdot) + \omega(V_0, \cdot) =\\
    =\omega(V_\mathcal{H}, \cdot) = -d\mathcal{H}.
\end{gather}
We will call Hamiltonian functions satisfying 
\begin{equation}\label{eq:admissible_cond}
    d\mathcal{H}(D) = 0.
\end{equation}
\textit{admissible} and by $\mathcal{A}(M)$ we will denote the space of functions that satisfy \eqref{eq:admissible_cond}. Since tangent and cotangent vector bundles are isomorphic condition \eqref{eq:admissible_cond} is both necessary and sufficient for the existence of a Hamiltonian vector field of $\mathcal{H}$ \cite{bursztyn2013}. 
Therefore for a presymplectic $(M,\omega)$ only admissible Hamiltonian functions give Hamiltonian dynamics. This dynamics, in contrast to the symplectic case is not unique. The freedom of choosing Hamiltonian vector field corresponding to an admissible $\mc H$ grows with the dimension of the null distribution $D$. Finally, if a Hamiltonian function $\mathcal{H}$ is not admissible the notion of a corresponding Hamiltonian vector field and Hamilton's equations make no sense. 

\section{Quantum dynamics and Hamiltonian equations of motions}\label{sec:quant_mech}
In this section, we examine how the symplectic structure facilitates the interpretation of quantum dynamics within the framework of classical equations of motion.

Consider quantum system described by Hilbert space $\mathcal{S}\cong \mathbb{C}^d$. We define an equivalence relation on $\mathcal{S}$ such that:
\begin{equation}\label{eq:equiv_rel}
    \ket{\psi} \sim \ket{\phi} \Leftrightarrow \exists_{\lambda\neq 0} \; \ket{\psi} = \lambda \ket{\phi}.
\end{equation}
Each equivalence class of this relation is a line in $\mathcal{S}$. The set of all equivalence classes is called complex projective space $M = \mathbb{P}(\mathcal{S})$ and it is a $(2d-2)$-dimensional real ($(d-1)$-dimensional complex) differential manifold. Pure quantum states are elements of $M$. To abbreviate the notation we will denote elements of $M$ by their representatives that is simply as $\ket{\psi}$ or $\psi$ depending which is more convenient.

Any smooth curve in $M$ - $\ket{\psi(t)}$ can be represented as an action of a smooth map called evolution operator:
\begin{equation}
    \hat{U}: \mathbb{R} \ni t \mapsto \hat{U}(t) \in \mathbb U(d), \; \hat{U}(0) = \bbone,
\end{equation}
where $\mathbb U(d)$ is a unitary group, on some initial state $\ket{\psi_0}$. That is:
\begin{equation}
    \ket{\psi(t)} = \hat{U}(t)\ket{\psi_0}.
\end{equation}
On the other hand, for any unitary matrix $\hat{U}$ there exists hermitian matrix $\hat{H}$ such that:
\begin{equation}
    \hat{U} = e^{-i\hat{H}}.
\end{equation}
This leads to:
\begin{equation}\label{eq:deriv-hamil}
    \ket{\dot{\psi}} := \frac{d}{dt}\bigg|_{t=0} \ket{\psi(t)} = -i\hat{H}\ket{\psi_0},
\end{equation}
thus to any curve in $M$ we can assign at each point a Hermitian matrix that determines the tangent vector at the point. Note that this assignment is not unique. According to relation \eqref{eq:equiv_rel} states that are proportional to each other represent the same point in $M$. This implies that
\begin{enumerate}
    \item for any real $\phi(t)$ curves $\ket{\psi(t)}$ and $e^{i\phi( t)}\ket{\psi(t)}$ have the same tangent vector,\label{ambiguity1}
    \item the part of $\ket{\dot{\psi}}$ along $\ket{\psi_0}$ does not influence the tangent vector.
\end{enumerate}
We can easily get rid of the first ambiguity by assuming that $\hat{H}$ is traceless since unitaries generated by $\hat{H}$ and $\hat{H}' = \hat{H} + \phi \bbone$ differ only by a global phase $e^{i\phi}$. Thus we will identify the tangent space $T_\psi M$ with the set of traceless hermitian matrices $\mathcal{O}_0(\mathcal{S})$ keeping in mind that the second ambiguity implies
\begin{equation}\label{eq:ambiguity2}
    \left(\hat{F} - \hat{G} \right) \ket{\psi} \sim \ket{\psi} \Rightarrow V_{\hat{F}} = V_{\hat{G}},
\end{equation}
for any $\hat{F}$, $\hat{G}$ and related tangent vectors $V_{\hat{F}}$, $V_{\hat{G}}$.

In this notation the Fubini-Study symplectic form at $\psi$ is defined as \cite{kobayashi1996}
\begin{equation}
    \omega_\psi\left(V_{\hat{F}}, V_{\hat{G}}\right) = i \bra{\psi}[\hat{F}, \hat{G}] \ket{\psi} = i \langle [\hat{F}, \hat{G}] \rangle_\psi.
\end{equation}
Closeness of $\omega$ can be shown by straightforward computation (for example using Theorem 13 from \cite{spivak1999}). On the other hand, non-degeneracy comes from the observation that if for all $\hat{G}$
\begin{equation}
    \omega_\psi\left(V_{\hat{F}}, V_{\hat{G}}\right) = 0,
\end{equation}
then for $\rho_\psi = \ket{\psi}\bra{\psi}$ we have
\begin{gather}\label{eq:sympl_is_zero}
    0 = \omega_\psi\left(V_{\hat{F}}, V_{\hat{G}}\right) = \braket{[\hat{F}, \hat{G}]}_\psi = \\\nonumber
    = \mathrm{Tr}\left(\rho_\psi [\hat{F}, \hat{G}]\right) = \mathrm{Tr}\left([\rho_\psi, \hat{F}] \hat{G}\right)
\end{gather}
for all $\hat{G}$ which implies $[\rho_\psi, \hat{F}] = 0$. This is possible only when $\ket{\psi}$ is an eigenstate of $\hat{F}$. Thus by \eqref{eq:ambiguity2} the Hamiltonian $\hat F$ is equivalent to the matrix of zeros $\hat 0$ which tangent vector by \eqref{eq:deriv-hamil} is simply $0$. Hence we have $V_{\hat{F}} = V_{\hat{0}} = 0$.

Let us now consider what is the form of classical Hamiltonian for $\hat{H}$. First, note that the vector field $\psi \mapsto V_{\hat{H}}$, which for simplicity we will denote by $V_{\hat{H}}$, is Hamiltonian:
\begin{strip}
    \begin{gather*}
    d\left( \omega(V_{\hat{H}}, \cdot) \right)(V_{\hat{F}}, V_{\hat{G}}) = V_{\hat{F}}(\omega(V_{\hat{H}}, V_{\hat{G}})) 
    - V_{\hat{G}}(\omega(V_{\hat{H}}, V_{\hat{F}})) - \omega(V_{\hat{H}}, [V_{\hat{F}}, V_{\hat{G}}]) =\\
    = - \left\langle[\hat{F}, [\hat{G}, \hat{H}]] + [\hat{G}, [\hat{H}, \hat{F}]] + [\hat{H}, [\hat{F}, \hat{G}]] \right\rangle,
\end{gather*}
\end{strip}%
which by Jacobi identity is zero. Hence, there exists function $\mathcal{H}$ such that $d\mathcal{H} = -\omega(V_{\hat{H}}, \cdot)$. One possible choice of $\mathcal{H}$ is
\begin{equation}\label{eq:class-ham}
    \mathcal{H}(\psi) = \langle \hat{H} \rangle_\psi =\tr(\rho_\psi \hat H).
\end{equation}
Indeed, for any $V_{\hat F}$ we have
\begin{gather*}
    -d\mathcal{H}(V_{\hat{F}}) =\\
    =-\frac{d}{dt}\bigg|_{t=0} \bra{\psi} e^{i\hat{F}t} \hat{H} e^{-i\hat{F}t} \ket{\psi} =\\
    = i \langle [\hat{H}, \hat{F}] \rangle = \omega(V_{\hat{H}}, V_{\hat{F}}),
\end{gather*}
In fact for this choice we also have {\it equivariance}, that is 
\begin{equation}
    \mathcal{H}(U\psi) = \tr(\rho_{U\psi} \hat H)=\tr(U\rho_\psi U^\dagger \hat H).
\end{equation}
Concluding, one can view Schrödinger equation  with $\hat H$ as Hamilton equation with the Hamiltonian $\mc{H}$ given by \eqref{eq:class-ham}. 
\subsection{Relationship between local quantum and classical dynamics for composite systems}

The above construction becomes more interesting when applied to composite systems and entangled states. Consider a multipartite system
\begin{equation}\label{local-ham}
    \mathcal{S} = \mathcal{S}_{1} \otimes \ldots \otimes \mathcal{S}_L,
\end{equation}
and a {\it local} Hamiltonian that acts on every $\mathcal{S}_k$ separately
\begin{equation}\label{eq:sep_hamil}
    \hat{H} = \hat{H}_1 \otimes \bbone\otimes\ldots \otimes \bbone + \ldots \bbone \otimes\ldots\otimes\bbone \otimes   \hat{H}_L.
\end{equation}
Then note that for $\rho_\psi = \kb{\psi}{\psi}$ we have
\begin{gather}\label{admisible}
    \mathcal{H}(\psi) = \bra{\psi} \hat{H} \ket{\psi} = \mathrm{Tr}\left( \rho_\psi \hat{H} \right) =\\
    = \mathrm{Tr}\left(\rho_\psi^1 \hat{H}_1 \right) + \ldots +\mathrm{Tr}\left(\rho_\psi^L \hat{H}_L \right),
\end{gather}
where $\rho_\psi^k$ is the reduced density matrices of the subsystem $k$:
\begin{equation}
    \rho_\psi^k = \mathrm{Tr}_{\hat{k}}\rho_\psi
\end{equation}
and $\mathrm{Tr}_{\hat{k}}$ is the partial trace over all subsystems except subsystem $k$. Thus when $\hat H$ is local the values of the corresponding classical $\mathcal{H}$ depend only on the reduced density matrices. It is therefore natural to consider the map 
\begin{gather}\label{eq:momentum-map}
    \mu(\psi)=(\rho_\psi^1,\ldots,\rho_\psi^L),
\end{gather}
that assigns to a state $\psi$ the collection of its reduced density matrices. Let $K=\mathbb U(\mc S_1,\ldots,\mc S_L)$ be the local unitary group, that is: 
\begin{gather*}
    K=\{U_1\times\ldots\times U_L|\,U_k\in \mathbb U(\mc S_k)\}.
\end{gather*}
The set of all states that can be obtained from $\ket{\psi}$ using only local unitaries is the orbit, $O_\psi$, of $K$ passing through $\ket\psi$:
\begin{align}\nonumber
    O_\psi = \left\{ U\circ\ket{\psi} \;|\; U \in K\right\},
\end{align}
where $U\circ\ket{\psi}= U_1\otimes\ldots\otimes U_L\ket{\psi}$. Of course $O_\psi$ is a submanifold of $M$ and for any $\varphi \in O_\psi$ the tangent space $T_\varphi O_\psi$ is spanned by Hamiltonian operators of the form \eqref{eq:sep_hamil}. We will denote the set of all such Hamiltonian operators as $\mathcal{O}_0(\mathcal{S}_1, \ldots,\mathcal{S}_L)$. One easily checks that for $U\in K$ we have  \begin{gather}\label{equivariance}
\mu(U\circ\psi)=(\mathrm{Ad}_{U_1}\rho_\psi^1 ,\ldots,\mathrm{Ad}_{U_L}\rho_\psi^L)=\\=:\mathrm{Ad}_U\mu(\psi),
\end{gather}
where $\mathrm{Ad}_{U_k}\rho_\psi^k:=U_k\rho_\psi^kU_k^\dagger$ and
\begin{gather*}
    \mc{H}(U_1\otimes \ldots \otimes U_L\psi)=\\\nonumber \mathrm{Tr}\left(\mathrm{Ad}_{U_1}\rho_\psi^1 \hat{H}_1 \right) + \ldots +\mathrm{Tr}\left(\mathrm{Ad}_{U_L}\rho_\psi^L \hat{H}_L \right).
\end{gather*}

 Note that \eqref{equivariance} implies that $O_\psi$ is mapped by $\mu$ onto the { \it adjoint orbit} of $K$ through $\mu(\psi)$, that is
\begin{gather}
\mu(O_\psi)=O_{\mu(\psi)}:=\{\mathrm{Ad}_{U}\mu(\psi)|\,U\in K\}
\end{gather}
For any $\mu(\varphi) \in O_{\mu(\psi)}$ the tangent space $T_{\mu(\varphi)}O_{\mu(\psi)}$ is spanned by local Hamiltonians $\mathcal{O}_0(\mc{S}_1,\ldots,\mc{S}_L)$, that is for $\hat F\in \mathcal{O}_0(\mc{S}_1,\ldots,\mc{S}_L)$ the corresponding tangent vector is given by:
\begin{gather}\label{vtilde}
    \tilde{V}_{\hat F}=\frac{d}{dt}\bigg|_{t=0} (\mathrm{Ad}_{e^{-it\hat F_1}}\rho^1_\varphi,\ldots, \mathrm{Ad}_{e^{-it\hat F_L}}\rho^L_\varphi,)= \\\nonumber
    =i( [\rho^1_\varphi,, \hat{F}_1],\ldots, [\rho^L_\varphi, \hat{F}_L]).
    \end{gather}
It is well known that $O_{\mu(\psi)}$ is a symplectic manifold with the Kostant-Kirillov-Souriou symplectic form \cite{Kirillov} defined in our notation by
\begin{gather}
    \omega_{\mu(\varphi)}(V_{\hat F},V_{\hat G}):=\sum_{k=1}^L\tr(\rho_\varphi^k[\hat F_k,\hat G_k]).
\end{gather}
Moreover, for $\hat F, \hat G \in \mathcal{O}_0(\mathcal{S}_1,\ldots, \mathcal{S}_L)$ and $\ket{\varphi}  \in L_\psi$ we have
\begin{gather}\label{pullback}
    \omega_\varphi(V_{\hat F}, V_{\hat G}) =\mathrm{Tr}\left(\rho_\varphi [\hat{F}, \hat{G}]\right) =\\\nonumber
    = \sum_{k=1}^{L}\mathrm{Tr}\left(\rho_\varphi^k [\hat F_k, \hat G_k]\right) = \omega_{\mu(\varphi)}(\tilde{V}_{\hat F}, \tilde{V}_{\hat G})\\\nonumber
\end{gather}

\subsection{Partition of $M$ by spectra of reduced density matrices}
In summary, when focusing on local quantum evolutions, the corresponding classical dynamics is governed by a Hamiltonian function that depends on \(\mu(\psi)\) rather than \(\psi\). Furthermore, local quantum evolution preserves adjoint orbits. Thus, it is natural to partition \(M\) into disjoint subsets \(M_{\mu(\psi)}\) defined by 
\begin{gather}
M_{\mu(\psi)} := \mu^{-1}(O_{\mu(\psi)}),
\end{gather}
where \(M_{\mu(\psi)}\) represents the union of all local unitary orbits that map to the same adjoint orbit via \(\mu\). Seminal results from symplectic geometry \cite{Sjamaar} ensure that the spaces \(M_{\mu(\psi)}\) are connected for any state \(\psi\). These are referred to as stratified spaces; although \(M_{\mu(\psi)}\) may not be a manifold in itself, there exists an open and dense subset \(M_{\mu(\psi)}^0 \subset M_{\mu(\psi)}\) that forms a smooth manifold. Moreover, the spaces \(M_{\mu(\psi)}\) are generally not symplectic, and the restriction of \(\omega\) to \(M_{\mu(\psi)}^0\) may exhibit degeneracy. Therefore, in general, \(M_{\mu(\psi)}^0\) are presymplectic manifolds. Equation \eqref{admisible} indicates that the admissible Hamiltonians on \(M_{\mu(\psi)}\) are precisely those that depend on reduced density matrices, thereby being determined by Hamiltonians on \(O_{\mu(\psi)}\). For a state \(\varphi \in M_{\mu(\psi)}\), we define:
\begin{gather}
E(\varphi) := \text{dim} D(M_{\mu(\psi)}),
\end{gather}
where \(D(M_{\mu(\psi)})\) denotes the null distribution on \(M_{\mu(\psi)}^0\). In \cite{sawicki2011}, the relationship between the entanglement of \(\varphi \in M{\mu(\psi)}\) and \(E(\varphi)\) was established in the case of two-qubit systems (\(L=2\)). This scenario is particularly straightforward since \(M_{\mu(\psi)}\) coincides exactly with \(O_{\psi}\). In Section \ref{sec:bipartite}, we briefly review these findings, and in Section \ref{sec:multipartite}, we extend our analysis to the case of \(L>2\) qudits, where \(M_{\mu(\psi)}\) encompasses more than a single local unitary orbit. We demonstrate that \(\text{dim} D(M_{\mu(\psi)})\) is determined by the corresponding degeneracy \(\text{dim} D(O_{\mu(\psi)})\), where \(\psi\) is any state from \(M_{\mu(\psi)}^0\). This result follows from the method of symplectic reduction.
. 

\section{Bipartite case}\label{sec:bipartite}

In this section we shortly review results for $\mc{S}=\mc{S}_1\otimes \mc{S}_2$, where $\mc{S}_k\simeq\CC^d$ \cite{sawicki2011} and  add additional discussion regarding admissible hamiltonians. In this case $M_{\mu(\psi)}$ is a manifold that consists of only one orbit $O_\psi$. From equation \eqref{pullback} we see that if $\tilde{V}_{\hat F}=0$, for $\hat F\in \mathcal{O}_0(\mc{S}_1,\mc{S}_2)$, and the corresponding $V_{\hat F}\neq 0$ then the restriction of $\omega$ to $O_\psi$ is degenerate and $V_{\hat F}$ is the null vector field on $O_\psi$. Combining this observation with \eqref{vtilde} we get that if
\begin{equation}\label{eq:degeneracy_condition} [\rho_\varphi^1, \hat F_1] = [\rho_\varphi^2, \hat F_2] = 0
\end{equation}
we have $\omega_\varphi(V_{\hat F}, V) = 0$ for all $V \in T_\varphi L_\psi$, that is $V_{\hat F}$ is a null vector field. Note that \eqref{eq:degeneracy_condition} means that $\mathrm{Ad}_{e^{it\hat F}}\mu(\psi)=\mu(\psi)$, that is $\mathrm{Ad}_{e^{it\hat F}}$ stabilizes $\mu(x)$. Thus the null vector fields arise as evolutions generated by $\hat F\in \mathcal{O}_0(\mc{S}_1,\mc{S}_2)$ that stabilize $\mu(\psi)$ but do not stabilize $\psi$. This implies that the null distribution is regular and its dimension can be calculated as the difference between the dimesnions of $O_\psi$ and $O_{\mu(\psi)}$. Next, notice that for maximally entangled $\ket{\psi}$ reduced density matrices are proportional to identity thus \eqref{eq:degeneracy_condition} is satisfied for all elements of $\mathcal{O}_0(\mathcal{S}_A, \mathcal{S}_B)$ and admissible are only constant Hamiltonians. On the other hand, for separable $\ket{\psi}$ reduced density matrices are pure therefore \eqref{eq:degeneracy_condition} is satisfied only when $V_{\hat F} = 0$ (see \eqref{eq:sympl_is_zero} and discussion below). Separable $\ket \psi$ satisfies also:
\begin{equation}
    \rho_\psi = \rho_\psi^A \otimes \rho_\psi^B,
\end{equation}
thus all smooth functions on $O_\psi$ are admissible. In summary, the degree of degeneracy $E(\psi)$ is invariant under local transformations, attains its maximal value for $\ket \psi$ maximally entangled and is zero if and only if $\ket \psi$ is separable.
The value of $E(\psi)$ was computed in \cite{sawicki2011}. Let $p_1, ..., p_d$ be Schmidt coefficients of $\ket\psi$ that is:
\begin{equation}
    \ket \psi = \sum_{i=1}^d \sqrt{p_i} \ket i \otimes \ket i,
\end{equation}
in some basis $\ket i$. Additionally, let $\tilde p_1, ... \tilde p_r$ be unique values of non-zero Schmidt coefficients and $m_1, ..., m_r$ their respective numbers of occurrences. Then we have \cite{sawicki2011}: 
\begin{equation}
    E(\psi) = \sum_{i=1}^r m_i^2 - 1.
\end{equation}
Therefore, $E$ is higher for states that have many repeating non-zero Schmidt coefficients which tend to be also more entangled.

\section{Many qudits}\label{sec:multipartite}
In this sections we consider $L$-qudit system, $\mathcal{S}=\left(\CC^d\right)^{\otimes L}$, where $L\geq 3$ and we derive formula for $E(\cdot)$. 

For $M=\PP(\mc S)$, the image of the
momentum map, $\mu(M)$, is the collection of adjoint orbits of $K=\mathbb U(\mc S_1,\ldots,\mc S_L)$, where every $\mc S_k\simeq\CC^d$. Note that the reduced $1$-qudit density matrices can be
diagonalized by the adjoint action of an element of the
group $\mathbb U(\mc S_1,\ldots,\mc S_L)$. To make the situation unambiguous we impose a
concrete order on the resulting eigenvalues, for example nonincreasing.

Thus every adjoint orbit in $\mu(M)$ is uniquely determined by the point $(\bar{p}_1,\bar{p}_2,\ldots,\bar{p}_L)\in \RR_+^{(d-1)L}$, where $\bar{p}_k$ is nonincreasingly ordered spectrum of $\rho_\varphi^k$, and this way we obtain a map 
$$\Psi:M\rightarrow\RR^{(d-1)L}_+$$ that assigns
to a state $\ket\phi$ the point $(\bar{p}_1,\bar{p}_2,\ldots,\bar{p}_L)\in \RR_+^{(d-1)L}$. The set
$\Psi(M)\subset \RR^{(d-1)L}_+$ parametrizes all adjoint orbits in $\mu(M)$ and by the convexity theorem of the momentum
map \cite{Atiyah82,GS84,K84} is a convex polytope, referred to as the Kirwan polytope. For $\alpha\in \Psi(M)$ we introduce the so called {\it reduced space}
\begin{equation}
\Psi^{-1}(\alpha)/K := \{O_\varphi \;|\; \varphi\in \Psi^{-1}(\alpha)\}.
\end{equation}
By our construction, for any $\alpha\in \Psi(M)$ there is $\ket\varphi\in M$ such that $M_{\mu(\varphi)}=\Psi^{-1}(\alpha)$.
Moreover, the foundational results in symplectic geometry \cite{Sjamaar} ensure that the spaces \(\Psi^{-1}(\alpha)/K\) are symplectic in nature. This observation proves to be pivotal for the subsequent main result presented herein:
\begin{theorem}\label{thm:main}
Let $\ket{\varphi}\in M_{\mu(\psi)}^0$. Then $$E(\varphi)=\dim_\RR O_\varphi-\dim_\RR O_{\mu(\varphi)}.$$
\end{theorem}
\begin{proof}
   By our construction $M_{\mu(\psi)}^0$ is a manifold. It is well known that $M_{\mu(\psi)}^0/K$, that is the so-called symplectic reduction of $M_{\mu(\psi)}^0$, is a connected symplectic manifold \cite{Sjamaar}. Consider the restriction of the symplectic form $\omega$ to $M_{\mu(\psi)}^0$. The degeneracy of this restricted form is the dimension of the null distribution $D(M_{\mu(\varphi)}^0)$.
   Let  $O_{\varphi}$ be the $K$-orbit through $\ket\varphi\in M_{\mu(\psi)}^0$. Note that \cite{guillemin84}:
   \begin{gather}
    \left(T_\varphi O_{\varphi}\right)^{\perp\omega}=\mathrm{Ker}_\varphi(d\mu)
   \end{gather}
   Thus the only candidates for null vectors in $T_\varphi M_{\mu(\psi)}^0$ are vectors that 
   are tangent to the fibres of the momentum map $\mu$, that is vectors that generate dynamics in $M_{\mu(\psi)}^0$ for which the value of $\mu$ does not change.  Such vectors can be divided into two parts. The first part contains vectors that are  tangent to the $K$-orbit through $\ket\varphi$ and by the reasoning described in Section \ref{sec:bipartite} its dimension is given $\dim_\RR O_\varphi-\dim_\RR O_{\mu(\varphi)}$. The remaining part  of the space $\left(T_\varphi O_{\varphi}\right)^{\perp\omega}$ is a symplectic space because the symplectic reduction guarantees that $M_{\mu(\psi)}^0/K$ is a symplectic manifold.
   \end{proof}
Theorem \ref{thm:main} ensures that the degree of degeneracy of $\omega$ at any given state \( \ket{\varphi} \in M_{\mu(\psi)}^0 \) corresponds to the degree of degeneracy of \( \omega \) when restricted to the manifold of locally unitary equivalent states containing \( \ket{\varphi} \). 

Moreover, for any admissible $\mathcal{H}$ the existence of null vector fields implies that we cannot uniquely determine related to it dynamic $V_{\mathcal{H}}$, in particular the correct dynamic $V_{\hat H}$. In particular for maximally entangled states the space of admissible Hamiltonian functions consists of constant functions and all vector fields are null. Thus the ambiguity in choosing Hamiltonian vector field is the largest (see also section 4 of \cite{bengtsson2007}).

\section{Summary}

In this work we have established a symplectic framework for analysing quantum entanglement,
identifying the degeneracy of the restricted symplectic form as a structural indicator of
non-separability. This geometric approach not only recovers known results in bipartite
systems but also extends naturally to multipartite settings, offering a unifying perspective
on entanglement classification under local unitary actions.

The measure introduced here can, in principle, be evaluated for different partitions of a
multipartite system into subsystems. Different partitions correspond to different symmetry
groups and hence to different momentum maps, making this approach adaptable to a variety
of physical settings. Such a procedure could provide a practical method for distinguishing
between genuinely multipartite entanglement and various forms of partial correlations. Further work could investigate its behaviour under symmetry constraints, its asymptotic properties in large-system limits, and its relation to other
geometric or information-theoretic quantifiers of entanglement.

\section{Acknowledgments}
The research was funded from the Norwegian Financial Mechanism 2014-2021 with project registration
number 2019/34/H/ST1/00636.

\printbibliography
\end{document}